\documentclass[runningheads]{llncs}

\usepackage{amsmath,amssymb,mathtools}
\usepackage{stmaryrd}
\usepackage{xspace}
\usepackage{listings}
\usepackage{hyperref}

\newcommand{\Tid}{\mathit{Tid}} 
\newcommand{\Var}{\mathit{Var}} 

\newcommand{\Prev}{{\mathop{\mathbf{Prev}}}} 
\newcommand{\Since}{{\mathbin{\mathbf{Since}}}} 
\newcommand{\Always}{{\mathop{\mathbf{Always}}}} 
\newcommand{\Sometime}{{\mathop{\mathbf{Sometime}}}} 
\newcommand{\Knows}{{\mathop{\mathbf{K}}}} 

\newcommand{\Implies}{\rightarrow}
\newcommand{\Iff}{\leftrightarrow}
\newcommand{\True}{\top}

\newcommand{\Active}{\mathsf{act}} 
\newcommand{\At}{\mathsf{at}} 

\newcommand{\PrevA}[2]{\Prev_{#1}\,#2} 
\newcommand{\PreA}[2]{\mathsf{pre}_{#1}\,#2} 

\newcommand{\HB}{\mathbin{\prec}}

\newcommand{\code}[1]{\texttt{#1}}

\begin{document}

\title{Compositional Verification of Concurrency Using Past-Time Temporal Epistemic Logic}
\titlerunning{Past-Time Epistemic Reasoning for Concurrency}

\author{Hamed Nemati \and Mads Dam}

\institute{KTH Royal Institute of Technology, Stockholm, Sweden\\
\email{\{hnnemati, mfd\}@kth.se}}

\maketitle

\begin{abstract}
Shared-memory concurrency is notoriously difficult to reason about because each thread executes under \emph{interference} from other threads.
At the same time, many correctness arguments for classical algorithms are fundamentally \emph{epistemic}: a thread enters a critical region only when, from its local view, it can rule out that another thread is concurrently in that region.
We make such arguments explicit by introducing a past-time temporal epistemic logic interpreted over interleaving executions with perfect recall over local histories.
Past-time operators support ``since''-style reasoning, while epistemic modalities capture what a given thread can conclude from its own observation history.
We give a semantics and a sound proof system, instantiate the logic to a simple shared-memory language with instrumented read/write observations, and illustrate the approach on Peterson's mutual exclusion algorithm: an epistemic condition established at the loop-exit step, combined with ownership-based stability obligations, yields global mutual exclusion.
We additionally show how the logic embeds rely--guarantee reasoning as an \emph{epistemic stability-lifting} principle, yielding a parallel-composition meta-theorem that matches the classical compatibility obligations.
\end{abstract}

\keywords{Concurrency \and Epistemic logic \and Temporal logic \and Rely--guarantee \and Program verification}

\section{Introduction}
\label{sec:intro}

Correctness of shared-memory concurrent programs is notoriously subtle~\cite{OwickiGries76,Jones83,OHearn07}.
Even for safety properties, a proof must account for \emph{interference}: while one thread executes, other threads can modify the shared store and invalidate reasoning that would be sound in a sequential setting.
The enduring difficulty is to obtain arguments that are simultaneously \emph{local}, so that they scale to realistic code, and \emph{robust}, so that they imply global correctness properties such as mutual exclusion~\cite{Peterson81}, linearizability~\cite{HerlihyWing90}, or freedom from data races.

A recurring pattern in textbook correctness arguments is \emph{epistemic}:
threads make decisions based on what they can or cannot infer from their own actions and observations.
In Peterson's mutual exclusion algorithm~\cite{Peterson81}, for example, each thread enters the critical section only after it has ruled out that the other thread will also enter.
In lock-free code, a failed \emph{compare-and-swap} provides evidence that some other thread wrote a value~\cite{HerlihyShavit08}.
Yet, mainstream assertion languages for concurrent program logics are extensional~\cite{OwickiGries76,Jones83,OHearn07,JungSSSTBD15}: they speak about the current global state, but do not directly capture what a particular thread \emph{knows} given its partial view.
As a result, proofs often encode epistemic reasoning indirectly, e.g., via delicate global invariants or auxiliary (ghost) state.

In this paper we advocate a lightweight assertion language based on \emph{past-time temporal epistemic logic}~\cite{Fagin95Book,DBLP:conf/ijcai/HalpernM85,HalpernMoses90,LichtensteinPZ85}.
The epistemic modality $\Knows_A \varphi$ states that $\varphi$ holds in all runs consistent with thread $A$'s observations; we use a perfect-recall semantics where observations are represented as a prefix of $A$'s local history.
Past-time temporal operators (\emph{previously} and \emph{since}) allow us to talk about ordering and persistence of facts~\cite{Pnueli77,LichtensteinPZ85}.
Crucially, we derive a \emph{local-time} operator $\Prev_A\varphi$ that refers to the most recent moment before the current one at which $A$ executed a step.
This enables thread-centric specifications and proofs that abstract away arbitrarily many steps by other threads.

We instantiate the logic\footnote{Partial mechanization and verification of the base logic in Isabelle/HOL is available at \url{https://github.com/FMSecure/pttel-theory}} to an interleaving semantics for shared-memory programs where each step is executed by a single thread and performs at most one shared-memory access (compound guards are desugared into read micro-steps).
Atomic propositions in the instantiation describe control-flow locations and derived access markers.
On top of this, we present a sound, sequent-style proof system combining classical reasoning, a standard S5 basis for knowledge~\cite{Fagin95Book}, and fixed-point style rules for past-time temporal connectives~\cite{LichtensteinPZ85}.
We illustrate the approach on Peterson's algorithm by proving an epistemic fact about the loop-exit step (expressed using our derived pre-state operator): whenever thread \(i\) enters the critical section, it knows that the loop guard was false in the pre-state of its last step.
Combined with stability/ownership obligations, this yields global mutual exclusion; as a corollary, a thread in the critical section knows that the other is not.

Beyond this argument, we show how the logic can serve as an assertion language for \emph{compositional} reasoning.
In particular, we develop an epistemic variant of rely--guarantee reasoning~\cite{Jones83,VafeiadisParkinson07,DoddsFPV09} in which rely/guarantee conditions are expressed as past-time temporal constraints on environment and component steps, and thread-local knowledge assertions are connected to global invariants via stability obligations.
To summarize, we make the following contributions:
\begin{itemize}
  \item We formalize an interleaving semantics for shared-memory programs that makes local histories explicit and supports perfect-recall epistemic reasoning.
  \item We define a past-time temporal epistemic logic with a derived \emph{local-time} operator $\Prev_A$ and give a sound proof system.
  \item We instantiate the logic to a simple shared-memory language with derived read/write access markers and demonstrate the expressiveness via an epistemic proof sketch of Peterson's mutual exclusion.
  \item We show how the logic can be embedded into a rely--guarantee style proof method, and we use the running Peterson example to illustrate the resulting compositional reasoning obligations.
\end{itemize}

\section{Motivation and Running Example}
\label{sec:overview}

\paragraph{Epistemic view of interference.}
A thread reasons from a partial view: it can only observe its local state and the values it reads from shared memory.
In our semantics, the modality \(\Knows_A\phi\) means that \(\phi\) holds at all points that are indistinguishable to thread $A$ given its entire observation history.

\paragraph{Last-step facts and stability.}
Many local proofs have the following shape.
Thread $A$ executes a step, establishes a fact \(\phi\) (typically by reading shared state), and then the environment executes some number of steps.
To use \(\phi\) later, $A$ needs an argument that the environment cannot have falsified it.
We express ``the most recent $A$-step'' via an abbreviation \(\PrevA{A}{\phi}\) (defined in \S\ref{sec:logic}) and use it to write
\[
  \mathsf{est}_A(\phi) \;\triangleq\; \PrevA{A}{(\Knows_A\phi)}.
\]
The key rely--guarantee style reasoning principle we develop is:
``if $A$ established $\phi$ at its last step and $\phi$ is stable under environmental interference since then, then $\phi$ holds now.''

\paragraph{Running example:}
We use Peterson's mutual exclusion algorithm for two threads, $0$ and $1$, as our running example.
It uses shared variables \code{flag[0]}, \code{flag[1]} $\in\{0,1\}$ and \code{victim}$\in\{0,1\}$:

\begin{lstlisting}[language=C,caption={Peterson's algorithm for thread i (with j = 1-i).},label={lst:peterson}]
flag[i] = 1;
victim  = i;
while (flag[j] == 1 && victim == i) { /* spin */ }
/* critical section */
flag[i] = 0;
\end{lstlisting}

Our primary safety goal is mutual exclusion:
$
  \Always\,\neg(\mathsf{InCS}_0 \wedge \mathsf{InCS}_1)
$, 
where \(\mathsf{InCS}_i\) abbreviates ``thread $i$ is at its critical-section label''.
We also use the example to illustrate epistemic properties, e.g., how a thread can justify that a guard it evaluated remains meaningful under permitted interference.


\section{Program Model and Observations}
\label{sec:model}

We work with a simple interleaving model of shared-variable concurrency.
The intent is not to propose a new operational semantics, but to make precise the semantic objects over which temporal and epistemic formulas are interpreted.

\paragraph{States and runs.}
Fix a finite set of thread identifiers \(\Tid=\{1,\dots,n\}\).
A \emph{global state} \(s\) consists of:
\begin{enumerate}
  \item a shared store \(\sigma : \Var \to \mathit{Val}\) mapping shared variables to values, and
  \item for each thread \(A\in\Tid\), a local component \(\lambda_A\) containing a control location (program counter) and the values of its local variables.
\end{enumerate}
We write \(s.\sigma\) for the shared store and \(s.\lambda_A\) for thread \(A\)'s local component.

The operational semantics induces a labelled transition relation
\(s \xrightarrow{A} s'\), meaning that \(s'\) is obtained from \(s\) by executing one atomic step of thread \(A\).
We assume \emph{interleaving} (exactly one label per step) and consider infinite \emph{runs}
\(\pi = s_0 s_1 s_2 \cdots\)
where for each \(i\ge 0\) there exists a unique thread \(\mathsf{act}(i)\in\Tid\) with
\(s_i \xrightarrow{\mathsf{act}(i)} s_{i+1}\).
(Strictly, \(\mathsf{act}\) depends on the run; we leave the run implicit when it is clear from context.)
We further assume the standard \emph{locality} (thread-frame) property of shared-variable concurrency: if \(s \xrightarrow{A} s'\), then \(s'.\lambda_B = s.\lambda_B\) for every thread \(B\neq A\).
That is, a step of one thread leaves the local components of all other threads untouched (it may of course change the shared store and the acting thread's own local component).
This property is what makes a thread's observation insensitive to interleaved steps of others, and we rely on it below when relating histories across environment steps.

\paragraph{Scheduling predicate.}
We interpret \(A\ \Active\) at position \((\pi,i)\) as the fact that thread \(A\) performs the step from \(s_i\) to \(s_{i+1}\), i.e., \(\mathsf{act}(i)=A\).

\paragraph{Observations and perfect recall.}
A key modelling choice for epistemic logic is what each thread can observe.
We use the standard choice for shared-memory programs: thread \(A\) observes its local component.
Formally, the observation function is
\(\mathsf{obs}_A(s) \triangleq s.\lambda_A\).

Although \(\mathsf{obs}_A(s)\) exposes only \(A\)'s \emph{local} component, this is not a restriction.
In a standard small-step semantics, a read from shared memory updates a local register; hence the \emph{result} of the read becomes part of \(\lambda_A\), and therefore part of \(A\)'s observation history.
For example, in Peterson's algorithm (Listing~\ref{lst:peterson}), evaluating the loop guard reads \code{flag[j]} and \code{victim}, and those values are reflected in the control decision and (in an instrumented semantics) can be recorded in locals.
The logic therefore does not assume that a thread ``magically'' observes the whole shared store; it reasons from the same information a sequential proof would use: local state plus the values returned by reads.

\paragraph{Instrumentation of reads as observable locals.}
To make ``what was read'' explicit in the state, so that it can be mentioned in atomic predicates or used in epistemic postconditions, we assume a lightweight instrumentation of the small-step semantics.
Concretely, for each thread \(A\) and each shared variable \(x\) that \(A\) may read, we assume a distinguished local variable \(\mathsf{lr}^A_x\) (``\(A\)'s last-read register for \(x\)'') that is updated on every \(A\)-step that reads \(x\): if the pre-state shared store satisfies \(\sigma(x)=v\), then the read step sets \(\mathsf{lr}^A_x := v\) in the post-state (in addition to any program-local register updates).
Because \(\mathsf{lr}^A_x\) is part of \(\lambda_A\), the value returned by a read is (i) remembered under perfect recall, and (ii) directly addressable by state predicates such as \(\mathsf{lr}^A_x=v\).
This instrumentation does not add ``ghost information'' beyond what the operational semantics already provides to the executing thread; it only reifies read results as explicit local state.

\paragraph{Desugaring to single-access micro-steps.}
When a high-level command (e.g., a compound guard) reads multiple shared variables, we assume it is compiled or desugared into a sequence of micro-steps, each performing at most one shared-memory access and updating the corresponding last-read register \(\mathsf{lr}^A_x\).
The Peterson case study in \S\ref{sec:peterson} uses this standard desugaring.

\paragraph{Asynchrony and ``no global clock''.}
Because observations are histories of local states (and we do not assume a global clock is observable), a thread cannot in general distinguish whether \emph{extra} environment steps occurred while its local state stayed the same.
This stuttering-insensitivity is essential for modelling what a thread can safely know between its own steps, and it matches the standard asynchronous perfect-recall semantics used in the interpreted-systems literature.

To model \emph{perfect recall}, we compare points in the system by comparing entire observation histories.
Given a run \(\pi\) and time \(i\), let \(k_1<\cdots<k_m\) be the (possibly empty) sequence of indices strictly below \(i\) at which \(A\) acts, i.e., \(\mathsf{act}(k_r)=A\) for \(r=1,\dots,m\) and \(k_m<i\).
Define the history
\[
  \mathsf{Hist}_A(\pi,i) \triangleq \langle \mathsf{obs}_A(s_0),\ \mathsf{obs}_A(s_{k_1+1}),\ \ldots,\ \mathsf{obs}_A(s_{k_m+1}),\ \mathsf{obs}_A(s_i)\rangle,
\]
recording \(A\)'s initial observation, the observation in the post-state of each past \(A\)-step, and the current observation.
Two points \((\pi,i)\) and \((\pi',i')\) are \emph{indistinguishable} to \(A\), written \((\pi,i)\sim_A(\pi',i')\), iff
\(\mathsf{Hist}_A(\pi,i)=\mathsf{Hist}_A(\pi',i')\).
This is the standard epistemic accessibility relation for asynchronous systems with perfect recall; being defined by equality of histories, it is an equivalence relation.
By locality, if \(\mathsf{act}(i)\neq A\) then \(\mathsf{obs}_A(s_{i+1})=\mathsf{obs}_A(s_i)\) and the index \(i\) contributes no new \(A\)-step, so \(\mathsf{Hist}_A(\pi,i)=\mathsf{Hist}_A(\pi,i+1)\): a thread's information set is unchanged across steps it does not perform.

\paragraph{Atomic propositions.}
In addition to the scheduling predicate \(A\ \Active\), we interpret atomic propositions as predicates over global states.
Concretely, we allow atoms that test the values of shared variables (e.g., \(x=v\)) and atoms that test the control location of a thread (e.g., \(\At(A,\ell)\), read ``thread \(A\) is at location \(\ell\)'').
This keeps the logic close to program states while remaining agnostic to a particular instruction set.

\section{Past-Time Temporal Epistemic Logic}
\label{sec:logic}

This section defines the temporal--epistemic language used throughout the paper and gives its semantics over the program model of \S\ref{sec:model}.

\subsection{Syntax}
We assume a set \(\mathit{AP}\) of atomic state predicates \(p\) interpreted over global states (e.g., equalities on shared variables and control-location tests such as \(\At(A,\ell)\)).
Formulas are generated as follows and we use \(\vee\), \(\Implies\), and \(\Iff\) as abbreviations:
\[
\phi,\psi ::= p \;\mid\; A\ \Active \;\mid\; \neg\phi \;\mid\; \phi\wedge\psi \;\mid\; \Prev\,\phi \;\mid\; \phi\ \Since\ \psi \;\mid\; \Knows_A\,\phi .
\]

\subsection{State and transition formulas}
For interference reasoning it is convenient to separate \emph{state} properties from \emph{step} properties.
A \emph{state formula} is an \emph{extensional} predicate: it is built from atomic predicates \(p\) using Boolean connectives only (no \(\Prev\), \(\Since\), or \(\Knows\)).
Its truth depends only on the current global state \(s_i\) (shared store and local components).
In contrast, the scheduling atom \(A\ \Active\) is a predicate of the transition label \(\mathsf{act}(i)\) (which thread executes the \emph{next} step) and is therefore \emph{not} a pure state predicate under our model.
Formulas with \(\Knows_A\) are also \emph{intensional}: their truth depends on the entire \(\sim_A\)-information set induced by \(A\)'s observation history.
A \emph{transition formula} (or \emph{step predicate}) mentions the immediately preceding state using \(\Prev\).
For instance, \((x \neq \Prev\; x)\) is intended to hold exactly when the last transition changed the value of \(x\).
Strictly, \(\Prev\) is an operator on \emph{formulas}, so an equation between a term and its previous value is syntactic sugar: for terms \(t_1,t_2\) we read
\[
  t_1 = \Prev\; t_2 \;\;\triangleq\;\; \Prev\;\True \,\Implies\, \textstyle\bigvee_{v}\bigl(t_1 = v \,\wedge\, \Prev\,(t_2 = v)\bigr),
\]
i.e.\ ``if a predecessor exists, the current value of \(t_1\) equals the value of \(t_2\) in that state.''
The guard \(\Prev\;\True\) makes the equation \emph{vacuously true at the initial time}; this is the intended reading, and it is what makes \(\Always\)-guarded frame and ownership conditions (\S\ref{sec:rg}, \S\ref{sec:peterson}) satisfiable, since otherwise they would be falsified at time \(0\).
We write
\[
  \mathsf{chg}(x) \;\triangleq\; \neg\,(x = \Prev\; x),
\]
so that \(\mathsf{chg}(x)\) is false at time \(0\) and, at any time \(i>0\), holds exactly when the transition from \(s_{i-1}\) to \(s_i\) changed the value of \(x\); we build further step predicates on top of it below.

\paragraph{Convention.}
The atom \(A\ \Active\) is the \emph{outgoing} label of the current state: \((\pi,i)\models A\ \Active\) iff the step from \(s_i\) to \(s_{i+1}\) is performed by \(A\).
Readers familiar with an \emph{incoming}-label convention (in which a state is labelled by the step that produced it) should note that, in our setup, the incoming label at \(s_i\) is expressed as \(\Prev(A\ \Active)\).
This choice is made once and used consistently.

\subsection{Derived operators}
We use standard derived past-time operators:
\[
\Always\;\phi \;\triangleq\; \neg(\True\ \Since\ \neg\phi)
\qquad
\Sometime\;\phi \;\triangleq\; \True\ \Since\ \phi .
\]
Intuitively, \(\Always\;\phi\) means that \(\phi\) has held at every state in the run prefix up to the present, while \(\Sometime\;\phi\) means that \(\phi\) held at some earlier state (possibly the present).

\subsection{Last-step operator}
\label{sec:logic:laststep}
Rely--guarantee arguments often use facts established at a thread's most recent step.
To express this, we distinguish the thread scheduled for the \emph{next} transition from the thread that executed the \emph{previous} transition.
Define the abbreviation
\[
  \mathsf{after}_A \;\triangleq\; \Prev\,(A\ \Active),
\]
so that \(\mathsf{after}_A\) holds at time \(i\) iff \(i>0\) and the transition from \(s_{i-1}\) to \(s_i\) was performed by thread \(A\).
We then define the ``last \(A\)-step'' operator used in \S\ref{sec:overview}:
\begin{equation}
\label{eq:lastA}
\PrevA{A}{\phi} \;\triangleq\; \neg\mathsf{after}_A\ \Since\ (\mathsf{after}_A \wedge \phi).
\end{equation}
Thus \(\PrevA{A}{\phi}\) holds at \((\pi,i)\) iff \(\phi\) held at the most recent time \(j\le i\) such that \(\mathsf{after}_A\) held at \(j\) (i.e., the most recent state reached by an \(A\)-step), and \(\mathsf{after}_A\) did not hold at any \(k\) with \(j<k\le i\).
If thread \(A\) has not yet executed any step, then \(\mathsf{after}_A\) has never held, and \(\PrevA{A}{\phi}\) is false.

We also use the derived abbreviation \(\PreA{A}{\phi} \triangleq \PrevA{A}{(\Prev\,\phi)}\), which refers to the \emph{pre-state} of \(A\)'s last step.
In particular, if \(\mathsf{after}_A\) holds now, then \(\PreA{A}{\phi}\) is equivalent to \(\Prev\,\phi\).

\subsection{Write events and ``last occurrence'' reasoning}
Using \(\mathsf{after}_A\) and \(\mathsf{chg}(x)\) we define a coarse write event:
\[
  \mathsf{Write}_A(x) \;\triangleq\; \mathsf{after}_A \wedge \mathsf{chg}(x).
\]
Because we treat each program step as atomic, \(\mathsf{Write}_A(x)\) holds at time \(i\) exactly when the transition from \(s_{i-1}\) to \(s_i\) was an \(A\)-step that changed \(x\).

Finally, we occasionally need ``last occurrence'' reasoning for events other than ``last \(A\)-step''.
Generalizing \eqref{eq:lastA}, for any marker \(W\) we define
\[
  \mathsf{Last}_W(\phi) \;\triangleq\; \neg W\ \Since\ (W \wedge \phi).
\]
For example, \(\mathsf{Last}_{\mathsf{Write}_A(x)}(x=v)\) states that at the state reached by the most recent \(A\)-write to \(x\), the value of \(x\) was \(v\).

\paragraph{Happens-before.}
We use a ``happens-before'' abbreviation to express that one event was the most recent occurrence of its kind and was preceded by an occurrence of another:
\[
\phi\ \HB\ \psi \;\triangleq\; \neg\psi\ \Since\ \bigl(\psi \wedge \Prev\,\Sometime\,\phi\bigr).
\]
Thus \(\phi\ \HB\ \psi\) holds at \((\pi,i)\) iff there is a most recent time \(j\le i\) at which \(\psi\) held, and \(\phi\) held at some strictly earlier time \(k<j\).
This is the pattern we actually need: \(\psi\) is the event that just happened (no re-occurrence since), and \(\phi\) already happened before it.
We use it in \S\ref{sec:peterson} for statements of the form ``some \code{victim=\(i\)} write occurred before the most recent \code{victim=\(j\)} write.''

\subsection{Semantics}
A model consists of the set of runs induced by the operational semantics (\S\ref{sec:model}) together with the indistinguishability relations \(\sim_A\).
We define satisfaction \((\pi,i)\models \phi\) inductively:
\begin{itemize}
  \item \((\pi,i)\models p\) iff \(p\) holds of the global state \(s_i\).
  \item \((\pi,i)\models A\ \Active\) iff \(\mathsf{act}(i)=A\).
  \item \((\pi,i)\models \neg\phi\) iff not \((\pi,i)\models \phi\).
  \item \((\pi,i)\models \phi\wedge\psi\) iff \((\pi,i)\models\phi\) and \((\pi,i)\models\psi\).
  \item \((\pi,i)\models \Prev\,\phi\) iff \(i>0\) and \((\pi,i-1)\models \phi\).
  \item \((\pi,i)\models \phi\ \Since\ \psi\) iff there exists \(j\) with \(0\le j\le i\) such that \((\pi,j)\models \psi\) and for all \(k\) with \(j<k\le i\), \((\pi,k)\models \phi\).
  \item \((\pi,i)\models \Knows_A\,\phi\) iff for all \((\pi',i')\) with \((\pi,i)\sim_A(\pi',i')\), we have \((\pi',i')\models \phi\).
\end{itemize}

As \(\sim_A\) is defined by equality of histories, it is an equivalence relation; consequently \(\Knows_A\) satisfies the S5 principles: normality (the K axiom), truth (T, factivity), and positive and negative introspection (4 and 5).
Recall from \S\ref{sec:model} that, by locality, \(A\)'s history (hence its information set and its knowledge) is unchanged across steps it does not perform.

\section{A Sound Proof System}
\label{sec:proofsystem}

We present a deductive system for the fragment of the logic used in the sequel.
The system is intended to support paper-and-pencil proofs (and, ultimately, automation); it is not meant as a complete axiomatization of all validities of temporal--epistemic logic.
Soundness is with respect to the semantics of \S\ref{sec:logic}.

\subsection{Propositional core}
A sequent has the form \(\Gamma \vdash \phi\), where \(\Gamma\) is a finite set of formulas (assumptions) and \(\phi\) is a formula (conclusion).
We write \(\bigwedge\Gamma\) for the conjunction of all formulas in \(\Gamma\).

We assume a standard sound sequent-style proof system for propositional logic (conjunction, disjunction, and negation), including:
\[
\frac{\Gamma\vdash \phi \qquad \Gamma\vdash \phi\Implies\psi}{\Gamma\vdash \psi}\ (\textsc{mp})
\qquad
\frac{\Gamma,\phi\vdash\psi \qquad \Gamma,\neg\phi\vdash\psi}{\Gamma\vdash \psi}\ (\textsc{lemE})
\]
and the usual introduction/elimination rules for \(\wedge\) and \(\vee\).

\subsection{Past-time operators}
The following rules capture the defining fixpoint properties of \(\Prev\) and \(\Since\).

\paragraph{Previous.}
\[
\frac{\Gamma\vdash \phi}{\Prev\;\Gamma \vdash \Prev\;\True \Implies \Prev\;\phi}\ (\textsc{prev})
\]
where \(\Prev\;\Gamma \triangleq \{\Prev\;\psi \mid \psi\in\Gamma\}\).
The guard \(\Prev\;\True\) ensures that we are not at the initial time; without it, the derived form \(\vdash \Prev\,\phi\) would be unsound even when \(\vdash \phi\) (at time \(0\) there is no predecessor and every formula of the form \(\Prev\,\phi\) is false).
In practice one uses this rule in the form: if the current time is not initial (i.e., \(\Prev\;\True\) is assumed) and all assumptions in \(\Gamma\) held at the predecessor, then \(\phi\) held there too.

\paragraph{Since.}
\[
\frac{\Gamma\vdash \psi}{\Gamma\vdash \phi\ \Since\ \psi}\ (\textsc{sinceI$_1$})
\qquad
\frac{\Gamma\vdash \phi \qquad \Gamma\vdash \Prev(\phi\ \Since\ \psi)}{\Gamma\vdash \phi\ \Since\ \psi}\ (\textsc{sinceI$_2$})
\]
\[
\frac{\Gamma\vdash \phi\ \Since\ \psi \qquad \Gamma,\psi\vdash \chi \qquad \Gamma,\Prev(\phi\ \Since\ \psi),\phi\vdash \chi}{\Gamma\vdash \chi}\ (\textsc{sinceE})
\]
Rule \textsc{sinceI$_1$} corresponds to choosing the witness time \(j=i\).
Rule \textsc{sinceI$_2$} corresponds to extending an existing witness by one step while maintaining \(\phi\).
Rule \textsc{sinceE} is the elimination rule corresponding to a single \emph{unfolding} of \(\phi\Since\psi \equiv \psi \vee (\phi \wedge \Prev(\phi\Since\psi))\): it splits on whether the witness is at the current time (left case, hypothesis \(\psi\)) or strictly earlier (right case, hypotheses \(\phi\) and \(\Prev(\phi\Since\psi)\)).
It is \emph{not} an induction principle: its right premise supplies \(\Prev(\phi\Since\psi)\), not the inductive hypothesis \(\Prev\,\chi\), so it cannot by itself propagate \(\chi\) across unboundedly many steps.

\paragraph{Since-induction.}
For invariance reasoning we use the genuine least-fixpoint induction rule for \(\Since\):
\[
\frac{\Gamma\vdash \psi\Implies\chi \qquad \Gamma\vdash (\phi\wedge\Prev\,\chi)\Implies\chi}{\Gamma\vdash (\phi\ \Since\ \psi)\Implies\chi}\ (\textsc{sinceInd})
\]
\emph{Side condition:} every formula in \(\Gamma\) is past-hereditary, i.e.\ holds at a point only if it holds at all earlier points of the same run; \(\Always\)-formulas (and \(\Gamma=\emptyset\)) satisfy this.
The decisive difference from \textsc{sinceE} is that the second premise carries \(\Prev\,\chi\) as an inductive hypothesis, which is exactly what is needed to transport \(\chi\) along the ``in between'' interval of a \(\Since\) witness.
We will see in \S\ref{sec:rg} that the invariance Lemma~\ref{lem:inv-from-pres} is an instance of \textsc{sinceInd} (with \(\phi=\True\), \(\psi=\mathsf{Init}\wedge I\), and \(\chi=\Always\,I\)).

We define \(\Sometime\) and \(\Always\) as abbreviations (\S~\ref{sec:logic}); their usual derived rules follow.

\subsection{Epistemic operator}
The epistemic modality follows the standard S5 principles for each thread \(A\).
We present them as sequent rules.

\paragraph{Normality.}
\[
\frac{\Gamma\vdash \phi}{\Knows_A\Gamma \vdash \Knows_A\phi}\ (\textsc{K})
\]
where \(\Knows_A\Gamma \triangleq \{\Knows_A\psi \mid \psi\in\Gamma\}\).

\paragraph{Truth and introspection.}
\[
\frac{\Gamma\vdash \Knows_A\phi}{\Gamma\vdash \phi}\ (\textsc{T})
\qquad
\frac{\Gamma\vdash \Knows_A\phi}{\Gamma\vdash \Knows_A\Knows_A\phi}\ (\textsc{4})
\qquad
\frac{\Gamma\vdash \neg\Knows_A\phi}{\Gamma\vdash \Knows_A\neg\Knows_A\phi}\ (\textsc{5})
\]

\subsection{Soundness}
\begin{theorem}[Soundness]
\label{thm:soundness}
If \(\Gamma\vdash \phi\) is derivable, then \(\Gamma\models \phi\), i.e., for every model and every point \((\pi,i)\), if \((\pi,i)\models \bigwedge\Gamma\) then \((\pi,i)\models \phi\).
\end{theorem}

\begin{proof}[Proof sketch]
Soundness of the propositional rules is standard; for the remaining rules we argue directly from the semantics.
For \textsc{prev}, suppose \((\pi,i)\models \bigwedge\Prev\;\Gamma\) and \((\pi,i)\models \Prev\;\True\).
The latter gives \(i>0\), and \((\pi,i)\models \Prev\;\psi\) yields \((\pi,i-1)\models \psi\) for each \(\psi\in\Gamma\), so \((\pi,i-1)\models\bigwedge\Gamma\).
By soundness of the premise \(\Gamma\vdash\phi\) we get \((\pi,i-1)\models\phi\), hence \((\pi,i)\models \Prev\;\phi\).
Rules for \(\Since\) are sound by unfolding the semantic definition: \textsc{sinceI$_1$} takes the witness time \(j=i\); \textsc{sinceI$_2$} keeps the same witness time and extends the ``in between'' interval by one step while preserving \(\phi\); \textsc{sinceE} is the elimination from one unfolding (case split on whether the witness is at the current time or strictly earlier).
For \textsc{sinceInd}, assume the premises and a point \((\pi,i)\) with \((\pi,i)\models\bigwedge\Gamma\) and \((\pi,i)\models \phi\Since\psi\); let \(j\le i\) be a witness, so \((\pi,j)\models\psi\) and \((\pi,k)\models\phi\) for \(j<k\le i\).
We show \((\pi,m)\models\chi\) for \(j\le m\le i\) by induction on \(m\): the side condition makes \(\bigwedge\Gamma\) hold at every \(m\le i\), so at \(m=j\) the first premise gives \(\chi\), and for \(m{+}1\) the inductive hypothesis gives \((\pi,m)\models\chi\), hence \((\pi,m{+}1)\models\Prev\,\chi\), and with \((\pi,m{+}1)\models\phi\) the second premise gives \(\chi\); taking \(m=i\) yields \((\pi,i)\models\chi\).
Rules \textsc{K}, \textsc{T}, \textsc{4}, and \textsc{5} are sound because each \(\sim_A\) is an equivalence relation, and \(\Knows_A\) is interpreted as universal quantification over \(\sim_A\)-accessible points.
\end{proof}

\section{Rely--Guarantee Style Reasoning via Stability}
\label{sec:rg}

This section shows how a rely--guarantee style reasoning pattern emerges naturally in past-time epistemic logic.
The key ingredients are: (i) a way to talk about the state reached by a thread's last step (\(\PrevA{A}{\cdot}\), \S\ref{sec:logic}), and (ii) a logical characterization of \emph{stability} under environment steps.

\subsection{Step predicates and stability}
Because the logic has only past-time operators, we express constraints on the \emph{most recent} transition using \(\Prev\) and the tag \(\mathsf{after}_A \triangleq \Prev(A\ \Active)\).
For example, the assertion ``if the last step was by thread \(B\), then shared variable \(x\) did not change'' is expressed as
\[
  \Always\bigl(\mathsf{after}_B \Implies (x = \Prev\; x)\bigr).
\]

\begin{definition}[Stability under the environment]
\label{def:stability}
Let \(A\in\Tid\) and let \(\phi\) be a state formula (in the extensional fragment of \S\ref{sec:logic}).
We say that \(\phi\) is \emph{stable under the environment of \(A\)} if
\[
  \mathsf{Stable}_A(\phi) \;\triangleq\; \Always\bigl(\neg\mathsf{after}_A \Implies (\Prev\;\phi \Implies \phi)\bigr).
\]
\end{definition}

Intuitively, \(\mathsf{Stable}_A(\phi)\) means that any step not performed by \(A\) preserves \(\phi\).

\subsection{Algebra of stable assertions}
Stability is the logical form of the familiar rely--guarantee side condition ``\(P\) is stable under the rely''.
A useful feature of Definition~\ref{def:stability} is that stability can be manipulated propositionally, which makes stability checks modular.

\begin{lemma}[Closure properties]
\label{lem:stable-closure}
For any thread \(A\) and state formulas \(\phi,\psi\):
\begin{enumerate}
  \item If \(\mathsf{Stable}_A(\phi)\) and \(\mathsf{Stable}_A(\psi)\), then \(\mathsf{Stable}_A(\phi\wedge\psi)\).
  \item If \(\mathsf{Stable}_A(\phi)\) and \(\mathsf{Stable}_A(\psi)\), then \(\mathsf{Stable}_A(\phi\vee\psi)\).
\end{enumerate}
\end{lemma}

\begin{proof}[Proof sketch]
Unfold Definition~\ref{def:stability}.
For (1), use that \(\Prev(\phi\wedge\psi)\) is equivalent to \((\Prev\;\phi)\wedge(\Prev\;\psi)\): on an environment step, \(\Prev(\phi\wedge\psi)\) implies both \(\Prev\;\phi\) and \(\Prev\;\psi\), so the two stability assumptions yield \(\phi\) and \(\psi\) now.
For (2), use that \(\Prev(\phi\vee\psi)\) implies \((\Prev\;\phi)\vee(\Prev\;\psi)\) and split cases.
\end{proof}

Note that stability is not closed under logical consequence in general; a short counterexample is given in Appendix~\ref{app:nf:stability}.

\paragraph{Frame conditions yield stability.}
A very common rely condition is a \emph{frame} property saying that the environment does not change some variable \(x\):
\[
  \mathsf{Frame}_A(x) \triangleq \Always\bigl(\neg\mathsf{after}_A \Implies (x=\Prev\; x)\bigr).
\]
From \(\mathsf{Frame}_A(x)\) we can derive \(\mathsf{Stable}_A(x=v)\) for any value \(v\), and by Lemma~\ref{lem:stable-closure} we can build stable assertions about tuples of frame-protected variables.
In \S\ref{sec:peterson} we exploit exactly this pattern for the ownership assumptions on \code{flag[i]}.

\subsection{From last-step facts to current facts}
The last-step operator \(\PrevA{A}{\phi}\) is designed so that, if it holds at time \(i\), then there exists a unique time \(j\le i\) that is the most recent state reached by an \(A\)-step, and \(\phi\) held at \(j\).
Between \(j\) and \(i\), all steps are by the environment (i.e., \(\neg\mathsf{after}_A\) holds on each intermediate state).

\begin{lemma}[Stability lifting]
\label{lem:stability-lift}
For any thread \(A\) and any state formula \(\phi\),
\[
  \mathsf{Stable}_A(\phi)\ \wedge\ \PrevA{A}{\phi}\ \models\ \phi.
\]
\end{lemma}

\begin{proof}[Proof sketch]
Let \((\pi,i)\) satisfy \(\mathsf{Stable}_A(\phi)\wedge \PrevA{A}{\phi}\), and let \(j\le i\) be the witness time for \(\PrevA{A}{\phi}\) in the semantics of \(\Since\) (cf.\ \eqref{eq:lastA}).
By the definition of \(\PrevA{A}{\phi}\) we have \((\pi,j)\models \phi\), and \((\pi,k)\models \neg\mathsf{after}_A\) for all \(k\) with \(j<k\le i\).
If \(j=i\) the conclusion is immediate; otherwise we show \((\pi,k)\models\phi\) by induction on \(k=j,\dots,i\).
The base case \(k=j\) is given.
For the step, assume \((\pi,k)\models\phi\) with \(j\le k<i\); then \(k+1>0\), so \((\pi,k+1)\models\Prev\,\phi\), while \((\pi,k+1)\models\neg\mathsf{after}_A\) by the hypothesis on \(\PrevA{A}{\phi}\).
Stability then yields \((\pi,k+1)\models\phi\).
Taking \(k=i\) gives the claim.
\end{proof}

\subsection{What knowledge buys us: persistence without re-reading}
A purely temporal statement \(\PrevA{A}{\phi}\) says that \(\phi\) held at \(A\)'s last step, but it does \emph{not} say that \(A\) is entitled to use \(\phi\) as a justified local assumption.
The epistemic modality makes that justification explicit: if \(\PrevA{A}{(\Knows_A\phi)}\) holds, then \(A\) previously had enough observational evidence to rule out all alternatives where \(\phi\) fails.

A key consequence of our asynchronous perfect-recall semantics is that a thread's information set does not change between its own steps.
Therefore, if \(A\) knew something at its last step, then it still knows it now (even if it has not executed any further instructions).

\begin{lemma}[Knowledge persistence between \(A\)-steps]
\label{lem:know-persist}
For any thread \(A\) and any formula \(\phi\),
\[
  \PrevA{A}{(\Knows_A\phi)}\ \models\ \Knows_A\phi .
\]
\end{lemma}

\begin{proof}[Proof sketch]
Let \((\pi,i)\models \PrevA{A}{(\Knows_A\phi)}\), and let \(j\le i\) be the witness time, so that \(\mathsf{after}_A\wedge\Knows_A\phi\) holds at \(j\) and \(\neg\mathsf{after}_A\) holds at every \(k\) with \(j<k\le i\).
By the definition of \(\PrevA{A}{\cdot}\), no \(A\)-step occurs strictly between \(j\) and \(i\); hence, by locality (\S\ref{sec:model}), \(\mathsf{obs}_A\) is constant on that interval and no new \(A\)-step is recorded.
Consequently \(\mathsf{Hist}_A(\pi,i)=\mathsf{Hist}_A(\pi,j)\), i.e.\ \((\pi,i)\sim_A(\pi,j)\), so the \(\sim_A\)-class is the same at \(i\) and at \(j\).
Since \(\Knows_A\phi\) holds at \(j\), it holds at \(i\).
\end{proof}

Lemma~\ref{lem:know-persist} explains why stability conditions are unavoidable: a thread can carry forward
\emph{knowledge}, but to carry forward the \emph{truth} of a state formula \(\phi\) it must also argue that
the environment cannot falsify \(\phi\) between its own steps.

\paragraph{Objective stability vs.\ epistemic quantification.}
\(\mathsf{Stable}_A(\phi)\) is evaluated on the \emph{actual} run (every environment step preserves \(\phi\)), whereas \(\Knows_A\phi\) quantifies over \emph{all} points compatible with \(A\)'s history.
Under the closed-program convention adopted below, the set of those alternatives is restricted to runs satisfying the interface specification, which already includes the rely itself; as a result, Corollary~\ref{cor:epistemic-lift} draws the useful factivity conclusion (\(\phi\) now) without having to wrap the rely inside \(\Knows_A\).

\subsection{Epistemic rely--guarantee}
In practice, a thread often establishes facts via local reasoning about what it knows at its last step.
Combining Lemma~\ref{lem:stability-lift} with factivity of knowledge (the S5 axiom \(\Knows_A\phi\Implies\phi\)) yields the following derived rule.

\begin{corollary}[Epistemic stability lifting]
\label{cor:epistemic-lift}
For any thread \(A\) and any state formula \(\phi\),
\[
  \mathsf{Stable}_A(\phi)\ \wedge\ \PrevA{A}{(\Knows_A\phi)}\ \models\ \phi.
\]
\end{corollary}

\begin{proof}[Proof sketch]
Although \(\Knows_A\phi\) is not itself a state formula, factivity gives \(\Knows_A\phi\Implies\phi\) \emph{pointwise}, so from \(\PrevA{A}{(\Knows_A\phi)}\) we deduce \(\PrevA{A}{\phi}\) (apply factivity at the witness time \(j\) of the \(\Since\) semantics in \eqref{eq:lastA}).
Lemma~\ref{lem:stability-lift}, applied to the state formula \(\phi\), then yields \(\phi\) at the current time.
\end{proof}

This corollary can be read as a rely--guarantee principle:
``
if thread \(A\) established \(\phi\) at its last step (by knowing it then), and \(\phi\) is stable under the environment since that step, then \(\phi\) holds now.
''

\subsection{Encoding rely and guarantee conditions}
A classic rely--guarantee proof assigns to each thread \(A\) a \emph{guarantee} describing what \(A\) may do in one step and a \emph{rely} describing what the environment may do.
In our setting, a guarantee for \(A\) is naturally expressed as a formula of the shape
\[
  \mathsf{G}_A \;\triangleq\; \Always\bigl(\mathsf{after}_A \Implies G_A^{\mathsf{step}}\bigr),
\]
where $G_A^{\mathsf{step}}$ is a predicate over the current state and the previous state (expressible using $\Prev$).
Similarly, a rely for \(A\) may be expressed as
\[
  \mathsf{R}_A \;\triangleq\; \Always\bigl(\neg\mathsf{after}_A \Implies R_A^{\mathsf{step}}\bigr).
\]
Compatibility of rely/guarantee assumptions becomes logical entailment obligations between these formulas (e.g., \(\mathsf{G}_B\models \mathsf{R}_A\) for \(B\neq A\) in the two-thread case); the resulting parallel-composition step is summarized in \S\ref{sec:rg:parallel}.
We use this encoding in \S\ref{sec:peterson} to phrase the standard ``ownership'' assumptions of Peterson's algorithm (only thread \(i\) writes \code{flag[i]}, and each thread writes \code{victim} only to its own identifier).

\paragraph{Convention: rely/guarantee as model constraints.}
Rely/guarantee (and ownership) clauses can be used in two natural ways in an epistemic setting, and in the sequel we commit to one of them.
\emph{(i) Closed-program reading (adopted here).}
We fix an interface specification \(\Phi\), typically a conjunction of \(\Always\)-guarded step constraints such as guarantees and ownership clauses, and restrict attention to runs that satisfy \(\Phi\).
Under this reading, \(\Knows_A\) quantifies only over alternatives that also satisfy \(\Phi\), so relies may be used inside knowledge without extra \(\Knows\)-wrapping; any formula that holds at all points of the restricted model is automatically known by every thread.
\emph{(ii) Open-system reading (not used in our development).}
Here the clause is an assumption about the actual run, added to the antecedent of an entailment but not restricting the epistemic alternatives.
Knowledge claims that depend on a rely then require an explicit hypothesis \(\Knows_A\mathsf{R}_A\).
The two readings are logically consistent and interchangeable for proofs that never appeal to knowledge about the rely itself; we adopt~(i) throughout.

\subsection{From rely--guarantee obligations to global invariants}
A standard use of rely--guarantee is to prove that some global safety property \(I\) is an \emph{invariant}.
In our past-time setting, the inductive step ``\(I\) is preserved by each transition'' is captured by the step predicate
\[
  \mathsf{Pres}(I) \;\triangleq\; \Always\bigl(\Prev\;\True \Implies (\Prev\; I \Implies I)\bigr),
\]
which says: whenever the current state has a predecessor, if \(I\) held there then \(I\) holds now.

\begin{lemma}[Invariant from step preservation]
\label{lem:inv-from-pres}
Let \(\mathsf{Init} \triangleq \neg\Prev\;\True\); this atom holds only at time \(0\), so it marks the initial state syntactically.
If \(I\) holds initially and \(\mathsf{Pres}(I)\) holds globally, then \(I\) has always held:
\[
  \Sometime(\mathsf{Init}\wedge I) \ \wedge\ \mathsf{Pres}(I)\ \models\ \Always\; I.
\]
\end{lemma}

\begin{proof}[Proof sketch]
Fix \((\pi,i)\) satisfying both premises; we show \((\pi,k)\models I\) for every \(k\le i\) by induction on \(k\).
For \(k=0\): \(\mathsf{Init}\) has unique witness \(0\), so \(\Sometime(\mathsf{Init}\wedge I)\) forces \((\pi,0)\models I\).
For the step, assume \((\pi,k)\models I\) with \(k<i\); then \((\pi,k+1)\models \Prev\,I\) and, since \(k+1>0\), also \((\pi,k+1)\models\Prev\,\True\), so the instance of \(\mathsf{Pres}(I)\) at \(k+1\) yields \((\pi,k+1)\models I\).
Taking \(k=i\) gives \((\pi,i)\models I\); as \(i\) was arbitrary, \(\Always\,I\) holds.

This is exactly the content of the \(\Since\)-induction rule \textsc{sinceInd} (\S\ref{sec:proofsystem}), instantiated with \(\phi=\True\), \(\psi=\mathsf{Init}\wedge I\), \(\chi=\Always\,I\), and \(\Gamma=\{\mathsf{Pres}(I)\}\) (an \(\Always\)-formula, so the side condition holds).
Its premises are discharged within the calculus: the base \(\Gamma\vdash(\mathsf{Init}\wedge I)\Implies\Always\,I\) holds because \(\mathsf{Init}\) forces the current time to be \(0\), where \(\Always\,I\) reduces to \(I\); and the step \(\Gamma\vdash(\True\wedge\Prev\,\Always\,I)\Implies\Always\,I\) follows since \(\Prev\,\Always\,I\) gives \(I\) at all earlier points while \(\mathsf{Pres}(I)\) extends it to the present.
The rule's conclusion is \(\Sometime(\mathsf{Init}\wedge I)\Implies\Always\,I\), i.e.\ the lemma.
(Note that the weaker unfolding rule \textsc{sinceE} does \emph{not} suffice here: its recursive premise would supply only \(\Prev\,\Sometime(\mathsf{Init}\wedge I)\), from which \(\mathsf{Pres}(I)\)---which needs \(\Prev\,I\)---cannot fire.)
\end{proof}

\paragraph{Compositional preservation checks.}
To establish \(\mathsf{Pres}(I)\) compositionally, we can split by which thread performed the last step.
Define
\[
  \mathsf{Pres}_A(I) \triangleq \Always\bigl(\mathsf{after}_A \Implies (\Prev\; I \Implies I)\bigr).
\]
By the interleaving semantics, at every time \(i>0\) exactly one \(\mathsf{after}_A\) holds, so \(\bigwedge_{A\in\Tid}\mathsf{Pres}_A(I)\) entails \(\mathsf{Pres}(I)\).
This is the direct analogue of the classic rely--guarantee proof obligation:
``
show that each thread preserves \(I\) on its own steps, under assumptions about how the environment may have affected the shared state.
''

\paragraph{A proof recipe.}
In the examples, we repeatedly use the following pattern.
\begin{enumerate}
  \item Choose a global invariant \(I\) expressing the safety goal.
  \item For each thread \(A\), prove a local step obligation \(\mathsf{Pres}_A(I)\) using sequential reasoning about \(A\)'s code, \emph{assuming} the rely constraints \(\mathsf{R}_A\).
  \item Discharge the rely assumptions by proving compatibility: the guarantees of the other threads imply \(\mathsf{R}_A\).
  \item Conclude \(\Always\; I\) by Lemma~\ref{lem:inv-from-pres}.
\end{enumerate}
The stability-lifting lemma (Lemma~\ref{lem:stability-lift}) is the workhorse for step (2) whenever the local proof needs to carry facts across unboundedly many environment steps, as in spin-wait loops.

\subsection{Parallel composition as a derived rule}
\label{sec:rg:parallel}
Classic presentations of rely--guarantee include an explicit \emph{parallel composition} rule.
In our setting, interleaving composition is built into the semantics, so the analogous step is a derived
meta-theorem stated directly over the trace predicates \(\mathsf{R}_A\) and \(\mathsf{G}_A\).

Write \(\mathsf{G}_{-A} \triangleq \bigwedge_{B\in\Tid\setminus\{A\}} \mathsf{G}_B\) for the environment guarantee
seen by thread \(A\).
For any invariant candidate \(I\), the following derived rule packages the usual rely--guarantee side conditions:
\[
\frac{
  \forall A\in\Tid.\ \ (\mathsf{G}_{-A} \models \mathsf{R}_A)
  \qquad
  \forall A\in\Tid.\ \ (\mathsf{G}_A \wedge \mathsf{R}_A) \models \mathsf{Pres}_A(I)
}{
  \bigl(\bigwedge_{A\in\Tid}\mathsf{G}_A\bigr) \models \mathsf{Pres}(I)
}
\]
The second premise is the model-relative local proof obligation: \(\mathsf{Pres}_A(I)\) is evaluated only at points where \(\mathsf{after}_A\) holds, so it combines sequential reasoning about \(A\)-steps (summarized by \(\mathsf{G}_A\)) with the rely assumptions.
In the two-thread case, this specializes to the familiar compatibility obligations \(\mathsf{G}_0\models\mathsf{R}_1\) and \(\mathsf{G}_1\models\mathsf{R}_0\), together with the local preservation checks \((\mathsf{G}_i \wedge \mathsf{R}_i)\models\mathsf{Pres}_i(I)\).
Combined with Lemma~\ref{lem:inv-from-pres}, this yields the classic rely--guarantee structure:
each thread is proved correct under a rely, and the parallel system is correct once each rely is justified by the
other threads' guarantees.

\section{Case Study: Peterson's Algorithm}
\label{sec:peterson}

We sketch how the temporal--epistemic and stability principles developed above can be used to structure a compositional proof of Peterson's mutual exclusion algorithm.

\subsection{Abstract control locations}
We assume each thread \(i\in\{0,1\}\) has control locations
\(\ell^\mathsf{flag}\), \(\ell^\mathsf{victim}\), \(\ell^\mathsf{waitF}\), \(\ell^\mathsf{waitV}\), \(\ell^\mathsf{cs}\), \(\ell^\mathsf{exit}\).
The two waiting locations encode the standard desugaring of the compound guard into read micro-steps (one shared-memory access per step): from \(\ell^\mathsf{waitF}\) a step reads \code{flag[j]} (where \(j=1-i\)) and enters \(\ell^\mathsf{cs}\) if it reads \(0\); otherwise it moves to \(\ell^\mathsf{waitV}\), reads \code{victim}, and enters \(\ell^\mathsf{cs}\) iff \(\code{victim}\neq i\) (else loops back).
We use the state predicate \(\At(i,\ell)\) to test the current location of thread \(i\), and we define
\[
  \mathsf{InCS}_i \;\triangleq\; \At(i,\ell^\mathsf{cs}).
\]

\subsection{Event predicates for key instructions}
To connect the logical proof to program steps, it is convenient to name a few derived \emph{event predicates}.
Because control locations are part of the state, we can recognize that a specific instruction has just executed by looking at the \emph{previous} location of the active thread.
For thread \(i\in\{0,1\}\), define:
\begin{align*}
  \mathsf{SetFlag}_i   &\;\triangleq\; \mathsf{after}_i \wedge \Prev\,\At(i,\ell^\mathsf{flag}), \\
  \mathsf{SetVictim}_i &\;\triangleq\; \mathsf{after}_i \wedge \Prev\,\At(i,\ell^\mathsf{victim}), \\
  \mathsf{EnterCS}_i   &\;\triangleq\; \mathsf{after}_i \wedge \At(i,\ell^\mathsf{cs}) \wedge \Prev\bigl(\At(i,\ell^\mathsf{waitF}) \vee \At(i,\ell^\mathsf{waitV})\bigr), \\
  \mathsf{ClearFlag}_i &\;\triangleq\; \mathsf{after}_i \wedge \Prev\,\At(i,\ell^\mathsf{exit}).
\end{align*}
These are all transition formulas in the sense of \S\ref{sec:logic}: each refers to the immediately preceding state via \(\Prev\).
Intuitively, \(\mathsf{SetFlag}_i\) marks the post-state of executing \code{flag[i]=1}, \(\mathsf{SetVictim}_i\) marks the post-state of executing \code{victim=i}, and \(\mathsf{EnterCS}_i\) marks the post-state of taking the loop-exit transition into the critical section.
We use these predicates both as \emph{markers} for last-occurrence reasoning (\S\ref{sec:logic}) and as convenient names for thread guarantees.
The operational semantics ensures that each step updates a single thread's location according to the control-flow graph.

\subsection{Ownership-style guarantees and derived stability}
We encode the standard ownership assumptions of Peterson's algorithm as step properties.

\paragraph{(G\(_\mathsf{flag}\)) Only thread \(i\) writes \code{flag[i]}.}
For \(i\in\{0,1\}\), define
\[
  \mathsf{OwnFlag}_i \;\triangleq\; \Always\bigl(\neg\mathsf{after}_i \Implies (\code{flag[i]} = \Prev\,\code{flag[i]})\bigr).
\]
This implies that facts about \(\code{flag[i]}\) are stable under the environment of thread \(i\).
In particular, we obtain \(\mathsf{Stable}_i(\code{flag[i]}=1)\) and \(\mathsf{Stable}_i(\code{flag[i]}=0)\) as instances of Definition~\ref{def:stability}.

\paragraph{(G\(_\mathsf{victim}\)) Threads write \code{victim} only to their own id.}
We express this at the level of control locations:
\[
  \mathsf{OwnVictim}_i \;\triangleq\;
  \Always\Bigl(\mathsf{after}_i \wedge \Prev\,\At(i,\ell^\mathsf{victim}) \Implies (\code{victim}=i)\Bigr).
\]
This says: if the most recent step was by \(i\) and that step executed the \code{victim=i} command, then the resulting state satisfies \(\code{victim}=i\).

\subsection{Rely/guarantee summary}
It is useful to make the rely/guarantee interfaces of Peterson explicit.
For each thread \(i\) (with \(j=1-i\)), we can view the following as a rely/guarantee pair over shared variables.

\paragraph{Guarantee of thread \(i\).}
Thread \(i\) never writes \code{flag[j]}, and only writes \code{victim} when executing \code{victim=i}:
\[
  \mathsf{G}_i^{\code{flag}} \triangleq \Always\bigl(\mathsf{after}_i \Implies (\code{flag[j]}=\Prev\,\code{flag[j]})\bigr),
\]
\[
  \mathsf{G}_i^{\code{victim}} \triangleq \Always\Bigl((\mathsf{after}_i \wedge \neg\mathsf{SetVictim}_i)\Implies (\code{victim}=\Prev\,\code{victim})\Bigr).
\]
(The deterministic postcondition \(\Always(\mathsf{SetVictim}_i\Implies \code{victim}=i)\) complements the second formula.)

\paragraph{Rely of thread \(i\).}
Dually, thread \(i\) relies on the environment (here, thread \(j\)) not writing \code{flag[i]}:
\[
  \mathsf{R}_i^{\code{flag}} \triangleq \Always\bigl(\neg\mathsf{after}_i \Implies (\code{flag[i]}=\Prev\,\code{flag[i]})\bigr),
\]
which coincides with the ownership-style assumption \(\mathsf{OwnFlag}_i\) in the two-thread setting.
For \code{victim}, \(i\) relies only on the weak fact that the environment changes it \emph{only} at \(\mathsf{SetVictim}_j\), captured by the frame condition introduced above.

\paragraph{Compatibility.}
The classic rely/guarantee compatibility check ``\(\mathsf{G}_j \models \mathsf{R}_i\)'' becomes a simple entailment between step predicates.
For instance, \(\mathsf{G}_j^{\code{flag}}\) entails \(\mathsf{R}_i^{\code{flag}}\): if \(j\) never writes \code{flag[i]} on its own steps and \(i\) and \(j\) are the only threads, then \code{flag[i]} is preserved on all non-\(i\) steps.
This is exactly the kind of modular interference argument that our stability operator \(\mathsf{Stable}_A(\cdot)\) is designed to encapsulate; a more detailed development would additionally record explicit frame conditions for the remaining shared variables.

\subsection{Local guarantees and what a thread can know}
Two kinds of facts are used in the standard Peterson argument:
(i) \emph{local} facts about what a thread just executed and what branch it took, and
(ii) \emph{shared-state} facts that must be shown stable under interference.
The temporal--epistemic view makes this split explicit.

\paragraph{Deterministic postconditions of atomic steps.}
From the sequential semantics of a single step we obtain the following guarantees:
\begin{align*}
  & \Always(\mathsf{SetFlag}_i \Implies \code{flag[i]}=1), \\
  & \Always(\mathsf{ClearFlag}_i \Implies \code{flag[i]}=0), \\
  & \Always(\mathsf{SetVictim}_i \Implies \code{victim}=i).
\end{align*}
These are not epistemic statements; they are plain step facts about the operational semantics.
In addition, we use the simple frame condition that \code{victim} changes only at these assignments:
\[
  \Always\Bigl(\neg(\mathsf{SetVictim}_0 \vee \mathsf{SetVictim}_1) \Implies (\code{victim} = \Prev\,\code{victim})\Bigr).
\]

\paragraph{Why knowing a shared fact requires a rely.}
Even if thread \(i\) just executed \code{flag[i]=1}, it does \emph{not} automatically follow that \(i\) knows \(\code{flag[i]}=1\) at the resulting state under our asynchronous perfect-recall semantics.
Intuitively, if the environment were allowed to change \code{flag[i]} without affecting \(i\)'s local state, then \(i\) cannot rule out that additional environment steps have already occurred since its assignment.
Thus \(\Knows_i(\code{flag[i]}=1)\) becomes valid only once we combine the local postcondition with an interference restriction such as \(\mathsf{OwnFlag}_i\).
Under the closed-program convention of \S\ref{sec:rg}, \(\mathsf{OwnFlag}_i\) is part of the model, so it already holds at every point quantified over by \(\Knows_i\); the knowledge statement then follows from the sequential postcondition of \(\mathsf{SetFlag}_i\) together with the fact that no environment step can falsify \code{flag[i]=1} in any admitted alternative run.

\paragraph{Knowing a past observation.}
In contrast, facts about \emph{what the thread observed at its own step} are robust: the past does not change.
Under the desugaring above, $\mathsf{EnterCS}_i$ holds exactly at the post-state of the \emph{read micro-step} that exits the spin loop. Such a step can enter $\ell^\mathsf{cs}$ only if it observed either $\code{flag[j]}=0$ or $\code{victim}\neq i$ in its pre-state, so the compound guard was false there.
Because the read result updates $i$'s local component (and thus its observation history), the same sequential justification yields an epistemic guarantee about that pre-state (using $\PreA{i}{\cdot}$ from \S\ref{sec:logic}):
\[
  \Always\Bigl(\mathsf{EnterCS}_i \Implies \Knows_i\,\PreA{i}{\neg(\code{flag[j]}{=}1 \wedge \code{victim}{=}i)}\Bigr).
\]
This is the form of knowledge we actually need: it records what \(i\) learned when it evaluated the guard, rather than asserting knowledge of the current shared-state values.

\subsection{Key epistemic/stability reasoning steps}
We highlight the two proof obligations where the temporal--epistemic view is most useful.

\paragraph{Carrying forward a local fact.}
When thread \(i\) executes \code{flag[i]=1}, the post-state satisfies \(\code{flag[i]}=1\) by the deterministic postcondition of the step.
Moreover, under the ownership constraint \(\mathsf{OwnFlag}_i\) the environment cannot change \code{flag[i]} without \(i\) acting; since \(\mathsf{OwnFlag}_i\) is part of the model, this fact is also something \(i\) can \emph{justify} at that step and keep using between its own steps (Lemma~\ref{lem:know-persist}), yielding \(\Knows_i(\code{flag[i]}=1)\) at the post-state of \(\mathsf{SetFlag}_i\).
Using Corollary~\ref{cor:epistemic-lift} we can transport this fact to later times:
\[
  \mathsf{OwnFlag}_i \wedge \PrevA{i}{\Knows_i(\code{flag[i]}=1)} \ \models\  \code{flag[i]}=1 .
\]
This is the simplest instance of the rely--guarantee pattern: the guarantee of the environment (it does not write \code{flag[i]}) provides the stability needed to keep a fact established at the last \(i\)-step valid while the environment executes.

\paragraph{The ``last entrant'' argument.}
We now make the textbook proof more explicit in the temporal--epistemic language.
Fix an arbitrary run \(\pi\) and suppose, for contradiction, that at some time \(t\) both threads are in the critical section:
\((\pi,t)\models \mathsf{InCS}_0 \wedge \mathsf{InCS}_1\).
Let \(i\in\{0,1\}\) be the thread that entered the critical section \emph{last}, and let \(j=1-i\).
Let \(e\le t\) be the entry point of \(i\), characterized by \((\pi,e)\models \mathsf{EnterCS}_i\) and \((\pi,e-1)\models \mathsf{InCS}_j\).

\paragraph{Step 1: \(j\) being in the critical section forces \(\code{flag[j]}=1\).}
By the control flow of thread \(j\), it executes \code{flag[j]=1} before reaching \(\ell^\mathsf{cs}\) and executes \code{flag[j]=0} only after leaving the critical section.
Under \(\mathsf{OwnFlag}_j\), only thread \(j\) writes \code{flag[j]}, so \(\code{flag[j]}=1\) is stable between \(j\)'s own flag steps; together these facts yield the global safety invariant
\[
  \Always(\mathsf{InCS}_j \Implies \code{flag[j]}=1).
\]
Applying this invariant at time \(e-1\), where \((\pi,e-1)\models \mathsf{InCS}_j\), gives \((\pi,e-1)\models \code{flag[j]}=1\), and hence at the entry step
\begin{equation}
\label{eq:pet-flagj-pre}
(\pi,e)\models \Prev(\code{flag[j]}=1).
\end{equation}

\paragraph{Step 2: entering the critical section records what \(i\) learned from the loop guard.}
By the epistemic guarantee for \(\mathsf{EnterCS}_i\) above, at time \(e\) thread \(i\) knows that in the \emph{pre-state} of its entry step the loop guard was false:
\[
  (\pi,e)\models \Knows_i\,\PreA{i}{\neg(\code{flag[j]}{=}1 \wedge \code{victim}{=}i)}.
\]
By factivity of knowledge (the T axiom, \S\ref{sec:proofsystem}) the \(\Knows_i\) modality may be discharged, giving \(\PreA{i}{\neg(\code{flag[j]}{=}1 \wedge \code{victim}{=}i)}\).
Since \(\mathsf{EnterCS}_i\) implies \(\mathsf{after}_i\), by the identity noted in \S\ref{sec:logic:laststep} the operator \(\PreA{i}{\cdot}\) coincides here with the ordinary previous-state operator, giving
\begin{equation}
\label{eq:pet-guard-pre}
(\pi,e)\models \Prev\bigl(\neg(\code{flag[j]}{=}1 \wedge \code{victim}{=}i)\bigr).
\end{equation}
Combining \eqref{eq:pet-flagj-pre} and \eqref{eq:pet-guard-pre} yields
\[
(\pi,e)\models \Prev(\code{victim}\neq i),
\]
and since there are only two threads this simplifies to
\begin{equation}
\label{eq:pet-victim-pre}
(\pi,e)\models \Prev(\code{victim}=j).
\end{equation}

\paragraph{Step 3: \eqref{eq:pet-victim-pre} forces \(j\)'s victim write to occur after \(i\)'s.}
Thread \(i\) executes \code{victim=i} immediately before entering the loop, i.e., \(\mathsf{SetVictim}_i\) occurs strictly before \(\mathsf{EnterCS}_i\).
By the deterministic postcondition of \code{victim=i}, the post-state of \(\mathsf{SetVictim}_i\) satisfies \(\code{victim}=i\), whereas \eqref{eq:pet-victim-pre} says that in the pre-state of the entry step \(\code{victim}=j\).
Hence \code{victim} must be written to \(j\) at some point strictly between \(\mathsf{SetVictim}_i\) and \(i\)'s entry.
By the frame condition above, the only steps that change \code{victim} are \(\mathsf{SetVictim}_0\) and \(\mathsf{SetVictim}_1\); and by the deterministic postconditions \(\mathsf{SetVictim}_0\Implies\code{victim}=0\) and \(\mathsf{SetVictim}_1\Implies\code{victim}=1\), the only one that can produce the value \(j\) is \(\mathsf{SetVictim}_j\) (the write \(\mathsf{SetVictim}_i\) yields \(\code{victim}=i\neq j\)).
Therefore some occurrence of \(\mathsf{SetVictim}_j\) lies strictly between \(\mathsf{SetVictim}_i\) and \(\mathsf{EnterCS}_i\), which we summarize using the happens-before abbreviation as
\[
  \mathsf{SetVictim}_i\ \HB\ \mathsf{SetVictim}_j .
\]

\paragraph{Step 4: the ordering contradicts that \(j\) was already in the critical section.}
Let \(t^\ast\le e\) be the time of the most recent \code{victim=j} write witnessed by Step 3, so that \(\mathsf{SetVictim}_i\) occurs strictly before \(t^\ast\); moreover \(t^\ast<e\), since \(\mathsf{SetVictim}_j\) is a \(j\)-step while \(\mathsf{EnterCS}_i\) is an \(i\)-step.
By the control flow of thread \(i\), executing \code{victim=i} is preceded by \code{flag[i]=1} and followed by the spin loop, with \code{flag[i]=0} executed only after leaving the critical section.
Hence at \(t^\ast\) thread \(i\) has executed \(\mathsf{SetFlag}_i\) and not yet any \(\mathsf{ClearFlag}_i\), and its most recent step at or before \(t^\ast\) lies in this window.
Every \(i\)-step in that window is either \(\mathsf{SetFlag}_i\) (deterministic postcondition \(\code{flag[i]}=1\)) or a later prelude/spin step that does not write \code{flag[i]}; combined with \(\mathsf{OwnFlag}_i\) (no environment step writes \code{flag[i]}), a short induction shows every such post-state satisfies \(\code{flag[i]}=1\).
Thus \(\PrevA{i}{(\code{flag[i]}=1)}\) holds at \(t^\ast\), and with \(\mathsf{Stable}_i(\code{flag[i]}=1)\) (derived from \(\mathsf{OwnFlag}_i\)) and Lemma~\ref{lem:stability-lift} we obtain \(\code{flag[i]}=1\) at \(t^\ast\).
Immediately after \(\mathsf{SetVictim}_j\) we also have \(\code{victim}=j\) by its deterministic postcondition, so thread \(j\)'s loop guard \(\code{flag[i]}=1 \wedge \code{victim}=j\) evaluates to true at \(t^\ast\).
Neither \code{flag[i]} nor \code{victim} changes between \(t^\ast\) and \(e\): by \(\mathsf{OwnFlag}_i\) only \(i\) writes \code{flag[i]}, and \(i\) remains in its spin loop through \(e\); and since \(t^\ast\) is the most recent \code{victim=j} write no later than \(e\), no later \(\mathsf{SetVictim}_j\) occurs before \(e\), while no \(\mathsf{SetVictim}_i\) can occur in \((t^\ast,e]\) (the spin loop contains no assignment to \code{victim}).
Consequently \(j\)'s guard remains true throughout \((t^\ast,e]\), so \(j\) cannot exit the spin loop in that interval; in particular \(\mathsf{InCS}_j\) fails at time \(e-1\), contradicting our assumption that \(j\) was already in the critical section then.

\paragraph{Summary.}
The formal structure mirrors the informal proof, but the logic forces a clean separation:
\emph{knowledge} is used only to justify the pre-state guard fact \eqref{eq:pet-guard-pre}, whereas \emph{stability} (rely/guarantee) transports \code{flag}-facts across unbounded environment interference. Moreover, once mutual exclusion is established as a global invariant, the advertised epistemic safety property follows as a corollary: since $\mathsf{InCS}_i$ is determined by $i$'s local control state, from $\Always\neg(\mathsf{InCS}_0\wedge\mathsf{InCS}_1)$ we obtain $\Always(\mathsf{InCS}_i\Implies\Knows_i\neg\mathsf{InCS}_j)$.

\paragraph{Discussion.}
The proof above is the familiar argument for Peterson's algorithm, but the temporal--epistemic presentation isolates two reusable reasoning patterns:
(i) carry forward a fact established at the last step using stability, and
(ii) separate what a thread can deduce at its own loop-exit point (an epistemic fact) from what remains invariant under environmental steps (a stability fact).
This separation is precisely what rely--guarantee proofs enforce by design.

A fully formal proof would add explicit predicates for ``thread \(i\) just executed line \(\ell\)'' and would discharge the local control-flow facts using an ordinary Hoare-style reasoning for the sequential fragment of each thread.
Our goal here is to show that the interference reasoning can be expressed and organized cleanly in the temporal--epistemic logic.

\section{Related Work}
\label{sec:related}

\paragraph{Compositional concurrency reasoning.}
Owicki--Gries introduced non-interference reasoning for proving assertions about parallel programs by checking that each thread preserves the assertions of the others~\cite{OwickiGries76}.
Jones' rely--guarantee method made environmental assumptions explicit and remains a foundational technique for modular interference reasoning~\cite{Jones83}.
These ideas have been refined and mechanized in modern frameworks, including concurrent separation logic~\cite{Brookes07,OHearn07,JungSSSTBD15} and rely--guarantee/separation-logic hybrids~\cite{VafeiadisParkinson07,Feng08LRG,DoddsFPV09,Dinsdale-YoungDGPV10,PintoDG14,Dinsdale-YoungBGPY13}.
Our contribution is orthogonal to these program logics: we use past-time epistemic assertions to make thread-local inference explicit, and we package interference control as stability obligations (\S\ref{sec:rg}), mirroring the classical ``stable under rely'' check.

\paragraph{Epistemic logic and verification.}
We adopt the interpreted-systems semantics of knowledge with asynchronous perfect recall~\cite{Fagin95Book,HalpernMoses90}.
Epistemic perspectives on concurrent computations and proof systems for parallel processes were explored early on~\cite{HoekMeyer92,HoekMeyer94,vanHulstMeyer96,Knight13}.
Temporal--epistemic logics have a substantial automated-verification literature, with tool support in MCK and MCMAS~\cite{GammieMeyden04,LomuscioQuRaimondi09}; see also the surveys~\cite{LomuscioPenczek12chapter}.
Knowledge-based specifications are also common in information-flow security, including concurrent settings~\cite{BalliuDG11}.

\paragraph{Knowledge-based correctness characterizations.}
Knowledge has been used to characterize consistency conditions such as sequential consistency~\cite{Lamport79} and linearizability~\cite{HerlihyWing90}; see, e.g.,~\cite{GleissenthallRybalchenko13,Hirai10}.
Those works use knowledge to specify properties of executions as observed by clients or groups of observers, whereas we use (single-thread) knowledge as a proof principle inside compositional safety arguments.

\section{Conclusion}
\label{sec:conclusion}

We developed a past-time temporal epistemic logic for reasoning about shared-variable concurrent programs under an interleaving semantics with perfect recall.
The central technical point is that epistemic statements about what a thread knew at its last step can be lifted to current-state facts whenever those facts are stable under environmental interference.
This yields a clean rely--guarantee style reasoning principle formulated directly in the logic and illustrated on Peterson's algorithm.

\paragraph{Current scope and limitations.}
Our semantic instantiation uses sequentially-consistent interleaving with single-access micro-steps.
This is a deliberate baseline: it isolates the interaction between observation-based knowledge and interference control, but it does not yet address read--modify--write atomics (such as compare-and-swap), fences, or weak-memory reorderings.
Likewise, we interpret formulas over infinite runs; terminating threads can be accommodated by standard stuttering conventions (e.g., self-loops at a final control location), but we have not developed dedicated proof rules for termination or liveness.

\paragraph{Future work.}
Several directions suggest themselves.
On the mechanization side, extending the existing Isabelle/HOL development to cover the rely--guarantee meta-theorem of \S\ref{sec:rg:parallel} would provide a stronger foundation for case studies.
On the automation side, integrating the normal-form transformation of Appendix~\ref{app:normalform} with SMT- or BDD-based backends is a natural next step.
On the modelling side, extending the setting to richer forms of concurrency, notably weak memory models and fine-grained atomicity, seems a particularly promising direction, since the epistemic perspective is designed precisely around what a thread can soundly infer from partial observations.

\subsubsection*{\ackname}
We thank Robert K{\"u}nnemann and Yufei Liu for their contributions to the mechanization of the basic logic. This work has been supported by the Wallenberg AI, Autonomous Systems and Software Program (WASP) funded by
the Knut and Alice Wallenberg Foundation.

\bibliographystyle{splncs04}
\bibliography{ref}

\appendix
\section{Appendix: A Normal Form for Past-Time Formulas and a Remark on Stability}
\label{app:normalform}

This appendix collects two auxiliary observations referred to in the main text.
First (\S\ref{app:nf:core}--\ref{app:nf:epistemic}), we sketch a normal-form transformation for the pure past-time fragment (without knowledge); the goal is to support future automation efforts, and it is not required for the soundness results in the main text.
Second (\S\ref{app:nf:stability}), we record a counterexample showing that stability is not monotone with respect to logical consequence, as referenced in \S\ref{sec:rg}.

\subsection{Eliminating derived operators}
\label{app:nf:core}
Because \(\Always\phi\) and \(\Sometime\phi\) are defined from \(\Since\), they can be eliminated directly:
\[
\Always\,\phi \equiv \neg(\True\ \Since\ \neg\phi),
\qquad
\Sometime\,\phi \equiv \True\ \Since\ \phi.
\]

\subsection{Fixpoint unfolding for \texorpdfstring{\(\Since\)}{Since}}
\label{app:nf:since}
The operator \(\Since\) satisfies the standard unfolding equivalence
\[
\phi\ \Since\ \psi \;\equiv\; \psi\ \vee\ (\phi \wedge \Prev(\phi\ \Since\ \psi)).
\]
This equivalence is sound under our semantics because at time \(0\) the \(\Prev\) subformula is false, so the base case reduces to \(\phi\ \Since\ \psi \equiv \psi\) at the initial state, as expected.

\subsection{Negation normal form}
\label{app:nf:nnf}
For the past-time fragment, a standard negation-normal-form (NNF) transformation pushes negations to atomic predicates by using
\[
\neg\Prev\phi \equiv (\Prev\True \wedge \Prev\neg\phi)\ \vee\ \neg\Prev\True,
\]
and by rewriting \(\neg(\phi\ \Since\ \psi)\) via the unfolding above and De Morgan's laws.
In practice, one often avoids expanding \(\neg(\phi\ \Since\ \psi)\) eagerly, introducing instead auxiliary predicates and using induction/recursion (e.g., in a model checker).

\subsection{Comment on epistemic operators}
\label{app:nf:epistemic}
A full normal form for the temporal--epistemic logic would need to account for the semantics of \(\Knows_A\) (universal quantification over \(\sim_A\)-classes), which is orthogonal to the temporal unfolding.
One possible direction is to restrict to syntactic fragments where epistemic operators occur only at ``local time'' points (e.g., immediately after a thread step) and to combine the temporal unfolding with standard S5 reasoning.
We leave this to future work.

\subsection{Counterexample: stability is not closed under consequence}
\label{app:nf:stability}
The closure properties in Lemma~\ref{lem:stable-closure} hold for Boolean connectives, but stability is \emph{not} monotone w.r.t.\ logical consequence.
Fix a thread \(A\) and a shared variable \(x\) ranging over integers, and let \(\phi \triangleq (x=0)\) and \(\psi \triangleq (x=0 \vee x=1)\), so that \(\models (\phi \Implies \psi)\).
Consider a run in which thread \(A\) never acts and the environment performs a step that changes \(x\) from \(1\) to \(2\).
At that environment step, \(\Prev\,\psi\) holds but \(\psi\) does not, so \(\mathsf{Stable}_A(\psi)\) fails.
On the other hand, \(\mathsf{Stable}_A(\phi)\) holds vacuously on the same run if the antecedent \(\Prev\,\phi\) is never true on environment steps (e.g., \(x\) is never \(0\) at those points).
Hence \(\mathsf{Stable}_A(\phi)\) and \(\models(\phi\Implies\psi)\) do not imply \(\mathsf{Stable}_A(\psi)\).

\end{document}